%% file: main.tex
\documentclass[11pt]{article}
\input{header.tex}

\newcommand{\BPP}{\mathsf{BPP}}
\newcommand{\BPPZ}{\mathsf{BPP}^0}

\algnewcommand\algorithmicinput{\textbf{Input: }}
\algnewcommand\Input{\item[\algorithmicinput]}
\algnewcommand\algorithmicoutput{\textbf{Output: }}
\algnewcommand\Output{\item[\algorithmicoutput]}
\algnewcommand{\OneLineIf}[2]{
  \State \algorithmicif\ #1\ \algorithmicthen\ #2}

\definecolor{CommentColor}{RGB}{102,102,102}
\newcommand{\LComment}[1]{\Statex{{\color{CommentColor}\hspace{\algorithmicindent}{$\triangleright$ \textit{#1}}}}}

\title{Better Boosting of Communication Oracles, or Not}
\author{
Nathaniel Harms \\
EPFL
\and
Artur Riazanov \\
EPFL}

\begin{document}

\maketitle

\newcommand{\Eq}{\textsc{Eq}}
\newcommand{\R}{\mathsf{R}}
\newcommand{\RPDT}{\mathsf{RPDT}}
\newcommand{\D}{\mathsf{D}}
\newcommand{\Rpriv}{\mathsf{R}^{\mathrm{priv}}}
\newcommand{\DEQ}{\mathsf{D}^{\textsc{Eq}}}
\newcommand{\HD}[2]{\textsc{HD}_{#1}^{#2}}
\newcommand{\EQ}{\textsc{Eq}}

\begin{abstract}
Suppose we have a two-party communication protocol for $f$ which allows the parties to make queries
to an oracle computing $g$; for example, they may query an \textsc{Equality} oracle. To translate
this protocol into a randomized protocol, we must replace the oracle with a randomized subroutine
for solving $g$. If $q$ queries are made, the standard technique requires that we boost the error of
each subroutine down to $O(1/q)$, leading to communication complexity which grows as $q \log q$.
For which oracles $g$ can this na\"ive boosting technique be improved?

We focus on the oracles which can be computed by constant-cost randomized protocols,
and show that the na\"ive boosting strategy can be improved for the \textsc{Equality} oracle
but not the \textsc{1-Hamming Distance} oracle. Two surprising consequences are (1)
a new example of a problem where the cost of computing $k$ independent copies grows superlinear in
$k$, drastically simplifying the only previous example due to Blais \& Brody (CCC 2019); and (2) a
new proof that \textsc{Equality} is not complete for the class of constant-cost randomized
communication (Harms, Wild, \& Zamaraev, STOC 2022; Hambardzumyan, Hatami, \& Hatami, Israel Journal
of Mathematics 2022).
\end{abstract}

\thispagestyle{empty}
\setcounter{page}{0}
\newpage
\section{Introduction}

We typically require that randomized algorithms succeed with probability $2/3$, since the
probability can be boosted to any $1-\delta$ by taking a majority vote of $O(\log(1/\delta))$
repetitions. If many randomized subroutines are used within an algorithm, the probability of error
may accumulate, and one may apply standard boosting to each subroutine to bring the error
probability down to an acceptable level. We wish to understand when this is necessary, in the
setting of communication complexity.

Suppose two parties, Alice and Bob, wish to compute a function $f(x,y)$ on their respective inputs
$x$ and $y$, using as little communication as possible, and they have access to a shared (\ie
public) source of randomness.  A convenient way to design a randomized communication protocol to
compute $f(x,y)$ is to design a \emph{deterministic} protocol, but assume that Alice and Bob have
access to an \emph{oracle} (in other words, a subroutine) which computes a certain problem $g$
that itself has an efficient randomized protocol.

\begin{boxexample}
\label{example:trees}
The \textsc{Equality} problem is the textbook example of a problem with an efficient randomized
protocol \cite{KN96,RY20}: Given inputs $a, b \in [N]$, two parties can decide (with success
probability $3/4$) whether $a = b$, using only 2 bits of (public) randomized communication,
regardless of the domain size $N$. So, to design a randomized protocol for solving another problem
$f(x,y)$, we may assume that the two parties have access to an \textsc{Equality} oracle. For
example, suppose Alice and Bob have vertices $x$ and $y$ in a shared tree $T$, and wish to decide
whether $x$ and $y$ are adjacent in $T$. If $p(x)$ denotes the parent of $x$ in $T$, then Alice and
Bob can decide adjacency using two \textsc{Equality} queries: ``$x = p(y)$?'' and ``$y = p(x)$?''
\end{boxexample}

\begin{boxexample}
The \textsc{1-Hamming Distance} communication problem is denoted $\HD{1}{}$ and defined as
$\HD{1}{n}(x,y) = 1$ if $x,y\in\{0,1\}^n$ differ on exactly 1 bit, and $0$ otherwise. It has a
constant-cost randomized protocol, but unlike adjacency in trees, this protocol \emph{cannot} be
expressed as a deterministic protocol using the \textsc{Equality} oracle \cite{HHH22dimfree,HWZ22}.
\end{boxexample}

Using oracles makes the protocol simpler, and also makes it clearer how and why
randomness is used in the protocol, which provides more insight into randomized communication (see
\eg \cite{CLV19,HHH22dimfree,HWZ22,FHHH24} for recent work using oracles to understand randomized
communication). But when we replace the \emph{oracle} for $g$ with a randomized protocol for
$g$, we must compensate for the probability that the randomized protocol produces an incorrect
answer.  Write $\D^g(f)$ for the optimal cost of a deterministic communication protocol for $f$
using an oracle for $g$ (where the players pay cost 1 to query the oracle). Write $\R_\delta(f)$ for
the optimal cost of a randomized protocol for $f$ with error $\delta$. Then the inequality
\begin{equation}
\label{eq:intro-informal-badboosting}
  \forall f\;,\qquad \R_{\delta}(f)
    = O\left( \D^g(f) \cdot \R_{1/4}(g) \cdot \log\left( \frac{\D^g(f)}{\delta} \right)
      \right) 
\end{equation}
follows from standard boosting: if there are $q = \D^g(f)$ queries made by
the protocol in the worst case, then we obtain a randomized protocol by simulating each of the $q$
queries to $g$ using a protocol for $g$ with error $\approx \delta/q$, sending $\R_{\delta/q}(g) =
O( \R_{1/4}(g) \cdot \log(q/\delta))$ bits of communication for each query. But is it possible
to improve on this na\"ive bound? The main question of this paper is:

\begin{boxquestion}
\label{question:intro-main-question}
For which oracle functions $g$ can \cref{eq:intro-informal-badboosting} be improved?
\end{boxquestion}

We focus on the oracles $g$ which have \emph{constant-cost} randomized communication protocols, like
\textsc{Equality}. Randomized communication is quite poorly understood, with many fundamental
questions remaining open even when restricted to the surprisingly rich class of constant-cost
problems. Many recent works have focused on understanding these extreme examples of efficient
randomized computation; see \cite{HHH22dimfree,HWZ22,HHH22counter,EHK22,HHPTZ22,HZ24,FHHH24} and the
survey \cite{HH24}. And indeed some of these works \cite{HHH22dimfree,HH24} use
\cref{eq:intro-informal-badboosting} specifically for the \textsc{Equality} oracle. So this is a
good place to begin studying \cref{question:intro-main-question}.  Our main result is:

\begin{boxtheorem}[Informal; see \cref{thm:intro-eq-betterboosting,thm:no-boosting-for-hd1}]
\label{thm:intro-main-informal}
\cref{eq:intro-informal-badboosting} can be improved for the \textsc{Equality} oracle,
but it is tight for the \textsc{1-Hamming Distance} oracle.
\end{boxtheorem}

This has some unexpected consequences, described below, and also answers
\cref{question:intro-main-question} for all \emph{known} constant-cost problems.  

Every known constant-cost problem $g$ satisfies either $\D^\EQ(g) = O(1)$ or
$\D^g(\HD{1}{}) = O(1)$ (\cite{FGHH24} gives a survey of all known problems). Therefore we answer \cref{question:intro-main-question} for all \emph{known} constant-cost oracles.
Towards an answer for \emph{all} constant-cost oracles, we show that the technique which allows us
to improve \cref{eq:intro-informal-badboosting} works \emph{only} for the \textsc{Equality} oracle
(\cref{prop:no-generalization}).

\bigskip 
Our main proof also has two other surprising consequences:
\medskip

\noindent\textbf{Direct sums.} Direct sum questions ask how the complexity of computing $k$ copies
of a problem grows with $k$ (see \eg~\cite{FKNN95,KRW95,BBCR13,BB19}). Recently, \cite{BB19}
answered a long-standing question of \cite{FKNN95} by providing the first example of a problem where
the communication complexity of computing $k$ independent copies grows \emph{superlinearly} with
$k$. Their example is specially designed to exhibit this behaviour and goes through the
query-to-communication lifting technique. In our investigation of
\cref{question:intro-main-question}, we show that computing $k$ independent copies of the
drastically simpler, constant-cost \textsc{1-Hamming Distance} problem requires $\Omega\left(k
\log k\right)$ bits of communication (\cref{thm:no-boosting-for-hd1}). As a corollary, we also show a
similar direct sum theorem for randomized parity decision trees (\cref{cor:intro-rpdt}).\\

\noindent\textbf{Oracle separations.} In an effort to better understand the power of randomness in
communication, recent works have studied the relative power of different oracles.  \cite{CLV19} show
that the \textsc{Equality} oracle is not powerful enough to simulate the standard communication
complexity class $\BPP$ (\ie $N \times N$ communication matrices with cost $\poly\log\log N$), \ie
\textsc{Equality} is not \emph{complete} for $\BPP$.  \cite{HHH22dimfree,HWZ22} showed that
\textsc{Equality} is also not complete for the class $\BPPZ$ of \emph{constant-cost} communication
problems, because \textsc{1-Hamming Distance} does not reduce to it; and \cite{FHHH24} show that
there is \emph{no} complete problem for $\BPPZ$. There are many lower-bound techniques for
communication complexity, but not many lower bounds for communication with \emph{oracles}. Our
investigation of \cref{question:intro-main-question} gives an unexpected new proof of the separation
between the \textsc{Equality} and \textsc{1-Hamming Distance} oracles; our proof is
``algorithmic''\!\!, and arguably simpler than the Ramsey-theoretic proof of \cite{HWZ22} or the
Fourier-analytic proof of \cite{HHH22dimfree}.

\subsection*{Further Motivation, Discussion \& Open Problems}

Let's say a constant-cost oracle function $g$ has \emph{better boosting} if
\[
  \forall f \;:\qquad \R_\delta(f) = O( \D^g(f) + \log(1/\delta)) \,.
\]
We showed that among the \emph{currently-known} constant-cost oracle functions $g$, better boosting
is possible if and only if $\DEQ(g) = O(1)$, and we observed that among \emph{all} constant-cost
oracles, only the \textsc{Equality} oracle satisfies the properties used to prove
\cref{thm:intro-eq-betterboosting}.  So, permit us the following conjecture:

\begin{conjecture}
\label{conj:better-boosting-bppz}
An oracle function $g \in \BPPZ$ has better boosting if and only if $\DEQ(g) = O(1)$.
\end{conjecture}

To disprove this conjecture, we need a new example of a constant-cost (total) communication
problem that is not somehow a generalization of \textsc{$1$-Hamming Distance}. Such an example would
be very interesting, so in that regard we hope the conjecture is false.

\newcommand{\UPPZ}{\mathsf{UPP}^0}
One more motivation of the current study is to find an approach towards a question of \cite{HZ24}
about the intersection between communication complexity classes $\UPPZ \cap \BPPZ$, where $\UPPZ$
denotes the class of problems with bounded sign-rank, or equivalently, constant-cost
\emph{unbounded-error} randomized protocols \cite{PS86}. Writing $\mathsf{EQ}^0$ for the class of
problems $g$ where $\DEQ(g) = O(1)$, \cite{HZ24} asks:

\begin{question}[\cite{HZ24}]
\label{question:discussion-hz}
Is $\UPPZ \cap \BPPZ = \mathsf{EQ}^0$?
\end{question}

This question seems challenging;
as noted in \cite{HZ24}, a positive answer would imply other conjectures about $\UPPZ \cap \BPPZ$,
notably the conjecture of \cite{HHPTZ22} that \textsc{$1$-Hamming Distance} does not belong to
$\UPPZ$, which would be the first example of a problem in $\BPPZ \setminus \UPPZ$. \cite{HHPTZ22}
showed that all known lower-bound techniques against $\UPPZ$ fail to prove this.  But a positive
answer to \cref{question:discussion-hz} implies that all oracles in $\UPPZ \cap \BPPZ$ have better
boosting, so a weaker question is:

\begin{question}
\label{question:discussion-upp0}
Do all oracles in $\UPPZ \cap \BPPZ$ have better boosting?
\end{question}

Because of \cref{thm:intro-main-informal}, this weaker question would also suffice to prove that
\textsc{$1$-Hamming Distance} does not belong to $\UPPZ$.  It is not clear to us whether
\cref{question:discussion-upp0} is easier to answer than \cref{question:discussion-hz}. If the
answer to \cref{question:discussion-upp0} is negative (\ie there is an oracle in $\UPPZ \cap \BPPZ$
which does not have better boosting), then either \cref{conj:better-boosting-bppz} or
\cref{question:discussion-hz} is false.

Similar questions about probability boosting were studied recently for query complexity in
\cite{BGKW20} who focused on the properties of the \emph{outer} function $f$ of which allow for
better boosting to compute $f \circ g^{\otimes k}$ , whereas one may think of our oracles as the
\emph{inner} functions.  We may rephrase \cref{thm:intro-eq-betterboosting} as a ``composition
theorem'' which says that for any function $f\colon \zo^k \to \zo$, the composed function $f \circ
(\textsc{Eq})^{\otimes k}$ which applies $f$ to the result of $k$ instances of \textsc{Equality} has
communication cost
\begin{equation}
\label{eq:intro-composition}
  \R_\delta(f \circ (\textsc{Eq})^{\otimes k}) = O(\mathsf{DT}(f) + \log(1/\delta))
\end{equation}
where $\mathsf{DT}(f)$ is the decision-tree depth of $f$. We prefer the statement in
\cref{thm:intro-eq-betterboosting} because it more clearly differentiates between the
\emph{protocol} and the \emph{problem}. To see what we mean, consider taking $f$ to be the
\textsc{And} function; the immediate consequence of \cref{eq:intro-composition} is that
$\R_{1/4}(\textsc{And} \circ (\textsc{Eq})^{\otimes k}) = O(k)$, whereas the immediate consequence
of \cref{thm:intro-eq-betterboosting} is that $\R_{1/4}(\textsc{And} \circ (\textsc{Eq})^{\otimes
k}) = O(1)$ because this function can be computed using 1 \textsc{Equality} query. To get the same
result from \cref{eq:intro-composition} one must rewrite the \emph{problem} $\textsc{And} \circ
(\textsc{Eq})^{\otimes k}$ as a different decision tree over different inputs.

\section{Definitions: Communication Problems and Oracles}
\label{section:preliminaries}

We will use some non-standard definitions that are more natural for constant-cost problems. These
definitions come from \eg \cite{CLV19,HWZ22,HHH22dimfree,HZ24,FHHH24}.

It is convenient to define a \emph{communication problem} as a set $\cP$ of Boolean matrices,
closed under row and column permutations. The more standard definition has one fixed function $f\colon
\zo^n \times \zo^n \to \zo$ for each input size $n$, with communication matrix $M_f \in \zo^{2^n
\times 2^n}$, whereas we will think of a communication problem $\cP$ as possibly containing many
different communication matrices $M \in \zo^{N \times N}$ on each domain size $N$. (In the
adjacency-in-trees problem, \cref{example:trees}, there are many different trees on $N$ vertices,
which define many different communication matrices.)

For a fixed matrix $M \in \zo^{N \times N}$ and parameter $\delta < 1/2$, we write $\R_{\delta}(M)$
for the two-way, public-coin randomized communication complexity of $M$. For a communication problem
$\cP$, we write $\R_\delta(\cP)$ as the function
\[
  N \mapsto \max \left\{ \R_\delta(M) \;:\; M \in \cP, M \in \zo^{N \times N} \right\} \,.
\]
Then the class $\BPPZ$ is the collection of communication problems $\cP$ which satisfy
$\R_{1/4}(\cP) = O(1)$.

To define communication with oracles, we require the notion of a \emph{query set}:
\newcommand{\QS}{\mathsf{QS}}
\begin{boxdefinition}[Query Set]
A \emph{query set} $\cQ$ is a set of matrices closed under (1) taking submatrices; (2) permuting
rows and columns; and (3) copying rows and columns.  For any set of matrices $\cM$, we write
$\QS(\cM)$ for the closure of $\cM$ under these operations.
\end{boxdefinition}

Observe that if $\R_{1/4}(\cP) = O(1)$ then $\R_{1/4}(\QS(\cP)) = O(1)$, since constant-cost
protocols are preserved by row and column copying as well as taking submatrices.

\begin{boxdefinition}[Communication with oracles]
\label{def:oracle-communication}
Let $\cP$ be any communication problem, \ie set of Boolean matrices. For any $N \times N$ matrix $M$
with values in a set $\Lambda$, write $\D^\cP(M)$ for the minimum cost of a two-way deterministic
protocol computing $M$ as follows.  The protocol is a binary tree $T$ where each leaf node $v$ is
assigned a value $\ell(v) \in \Lambda$, and each inner node $v$ is assigned a query matrix $Q \in \zo^{N
\times N}$ where $Q \in \QS(\cP)$. On any pair of inputs $(i,j) \in [N] \times [N]$, the protocol
proceeds as follows: the current pointer $c$ is initiated as the root of $T$, and at every step, if
$Q_c(i,j) = 1$ then the pointer $c$ moves to its left child, and otherwise if $Q_c(i,j) = 0$ then
the pointer $c$ moves to the right.  Once the pointer $c$ reaches a leaf, the output of the protocol
is the value $\ell(c)$ assigned to the leaf $c$. It is required that $\ell(c) = M(i,j)$. The cost of the protocol is the depth of $T$.
\end{boxdefinition}

This definition differs from the standard definition of oracle communication because we do not
restrict the input size of the oracle. Specifically, each oracle query is represented by an $N
\times N$ matrix $Q \in \QS(\cP)$, obtained by taking a submatrix of an \emph{arbitrarily large}
instance of $P \in \cP$ and then copying rows and columns. This is the natural definition because
this preserves constant-cost randomized protocols, whereas preserving non-constant cost functions
usually requires restricting the size of the instance $P \in \cP$.

\begin{remark}
For constant-cost communication problems, \ie problems $\cP \in \BPPZ$, we will simply identify the
problem $\cP$ with its query set $\QS(\cP)$ since this does not change the communication complexity
of $\cP$.  For example, $\DEQ(\cdot)$ is $\D^\cQ(\cdot)$ where $\cQ$ is taken to be the closure
$\QS(\{I_{N,N}\})$ of the identity matrices.
\end{remark}


\section{Better Boosting of \textsc{Equality} Protocols}

We prove the first part of \cref{thm:intro-main-informal}, that \cref{eq:intro-informal-badboosting}
can be improved for the \textsc{Equality} oracle. This theorem will also be applied in the later
sections of the paper. 

The proof uses the ``noisy-search-tree'' argument of \cite{FPRU94}. This is a well-known idea that
was previously applied in \cite{Nis93} to get an upper bound on the communication complexity of
\textsc{Greater-Than}; see also the textbook exercise in \cite{RY20}. We only require the
observation that the argument works for arbitrary \textsc{Equality} queries, not just the binary
search queries used in those papers. Also, we did not find any complete exposition of the proof of
the \textsc{Greater-Than} upper-bound: the application of \cite{FPRU94} in \cite{Nis93} is black-box
and informal, and the models of computation in these two works do not match up, which causes some
very minor gaps in the proof\footnote{The gap is that the outputs of the \textsc{Equality}
subroutine are \emph{not} independent random variables. As far as we can tell, this very minor issue
persists in the textbook exercise in \cite{RY20} devoted to the \textsc{Greater-Than} problem.}, so
we make an effort to give a complete exposition here.

\textbf{Informal protocol sketch:}
The idea of the protocol is that an \textsc{Equality}-oracle protocol is a binary tree $T$, where
each node is a query to the oracle. On any given input, there is one ``correct'' path through $T$.
The randomized protocol keeps track of a current node $c$ in the tree $T$. In each round, the node
$c$ either moves down to one of its children, or, if it detects that a mistake has been made in an
earlier round, it moves back up the tree. There are two main ideas:
\begin{enumerate}
\item At every node $c$, the protocol can ``double-check'' the answers in \emph{all} ancestor nodes
with only $O(1)$ communication overhead, which \emph{implicitly} reduces the error of all previous
queries. This uses a special property of the \textsc{Equality} oracle, that a conjunction of
equalities $(a_1 = b_1) \wedge (a_2 = b_2) \wedge \dotsm \wedge (a_t = b_t)$ is equivalent to a
single equality $(a_1,a_2,\dotsc ,a_t) = (b_1, b_2, \dotsc, b_t)$. We can use this property to check
if the current node $c$ is on the ``correct'' path. (This simple observation is our contribution to
this argument.)
\item The random walk of the node $c$ through the tree is likely to stay close to the ``correct''
path; this is essentially the argument of \cite{FPRU94}.
\end{enumerate}

\begin{boxtheorem}
\label{thm:eq-betterboosting}
\label{thm:intro-eq-betterboosting}
For any $M \in \Lambda^{N \times N}$ with values in an arbitrary set $\Lambda$,
\[
  \normalfont\R_\delta(M) = O\left(\DEQ(M) +  \log\frac{1}{\delta} \right) \,.
\]
\end{boxtheorem}
\begin{proof}
Let $T$ be the tree of depth $d = \DEQ(M)$ as in
\cref{def:oracle-communication}. For a node $v$ in $T$ let $a_v,b_v\colon [N] \to \bN$ be the functions defining the oracle query at the node $v$ with $Q_v(i,j) = \EQ(a_v(i), b_v(j))$. Let $R \define 4\cdot\max\{d,
C\log(1/\delta)\}$ where $C$ is a sufficiently large constant, and construct a tree $T'$ by
replacing each leaf node $v$ of $T$ with another tree $L_v$ of depth $C \log(1/\delta)$ (where $C$ is a
sufficiently large constant), with each node $v'$ of $L_v$ being a copy of the parent node of
$v$ in $T$ (\ie the functions $a_{v'}, b_{v'}\colon [N] \to \bN$ are identical to those of the parent of
$v$). We then simulate the protocol defined by $T$ using \cref{alg:noisy-tree}.

\begin{algorithm}
\begin{algorithmic}[1]
\Input{Row $i$, column $j$ of communication matrix $M$.}
\State Initialize pointer $c \gets \mathrm{root}(T')$.
\For{$r \in [R]$}
  \State Let $P = (p_1, p_2, \dotsc, p_k)$ be the path in $T'$ from $\mathrm{root}(T')$ to $c$.
  \State Let $(q_1, q_2, \dotsc, q_t)$ be the subsequence of $P$ where the protocol has taken
    the left branch.
  \LComment{(\ie the nodes where the protocol previously detected ``equality''.)}
  \State Use the \textsc{Equality} protocol with error probability $1/4$ to check
    \[
      (a_{q_1}(i), a_{q_2}(i), \dotsc,  a_{q_t}(i))
      = (b_{q_1}(j), b_{q_2}(j), \dotsc, b_{q_t}(j)) ?
      \vspace{-1.3em}
    \]
    \label{ntcheckabove}
    \LComment{Re-check all previous ``equality'' answers simultaneously.}
    \If{inequality is detected on the sequence $q_1, \dotsc, q_t$}
      \LComment{A mistake was detected in an earlier round; go back up.}
      \State Update $c$ to be the parent of $c$ in $T'$.
    \Else
      \LComment{Check the current node and continue.}
      \State Use the \textsc{Equality} protocol with error probability $1/4$ to check
        $a_c(i) = b_c(j)$?
      \label{ntcheckcurrent}
      \OneLineIf{Equality is detected}{move $c$ to its left child, otherwise move $c$ to its right
child.}
    \EndIf
\EndFor
\If{$c$ belongs to a subtree $L_v$ (replacing leaf $v$ of $T$)}
\State{\textbf{return} $\ell(v)$. Otherwise \textbf{return} 0.}
\EndIf
\end{algorithmic}
\caption{Noisy-Tree Protocol}
\label{alg:noisy-tree}
\end{algorithm}
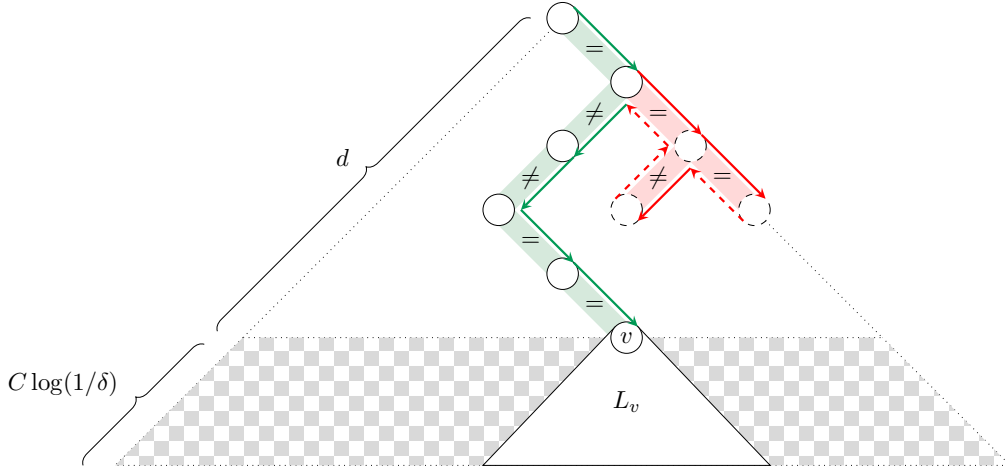
\begin{figure}[ht]
    \centering
    \scalebox{0.85}{
    \begin{tikzpicture}
        \draw[dotted] (0, 0) -- (-2,-2) -- (12, -2) -- (10, 0);
        \draw[dotted] (0, 0) -- (5, 5) -- (10, 0) -- cycle;
        \fill[pattern = {checkerboard}, pattern color = black!15] (0, 0) -- (-2,-2) -- (12, -2) -- (10, 0) -- cycle;

        \draw[decorate,decoration={brace,amplitude=5pt}] (-0.4, 0.1) --node[above left] {$d$~~~} (4.5, 5);
        \draw[decorate,decoration={brace,amplitude=5pt}] (-2.5, -2) --node[above left] {$C\log (1/\delta)$~~~} (-0.6, -0.1);
        \coordinate (root) at (5, 5) {};
        \coordinate (v1) at (6, 4) {};
        \coordinate (v1err1) at (7, 3) {};
        \coordinate (v1err2) at (8, 2) {};
        \coordinate (v1err3) at (6, 2) {};
        \coordinate (v2) at (5, 3) {};
        \coordinate (v3) at (4, 2) {};
        \coordinate (v4) at (5, 1) {};
        \coordinate (v5) at (6, 0) {};
        \draw[line width = 10pt, color = ForestGreen!15, line cap=round] (root) -- node[draw=none, color=black] {$=$} (v1);
        \draw[line width = 10pt, color = ForestGreen!15, line cap=round] (v1) -- node[draw=none, color=black] {$\neq$} (v2) --node[draw=none, color=black] {$\neq$}  (v3);
        \draw[line width = 10pt, color = red!15, line cap=round] (v1) -- node[draw=none, color=black] {$=$} (v1err1) --node[draw=none, color=black] {$=$}  (v1err2);
        \draw[line width = 10pt, color = red!15, line cap=round] (v1err1) -- node[draw=none, color=black] {$\neq$} (v1err3);
        \draw[line width = 10pt, color = ForestGreen!15, line cap=round] (v3) --node[draw=none, color=black] {$=$}  (v4) --node[draw=none, color=black] {$=$}  (v5);
        
        \draw[fill=white]  (root) circle (7pt);
        \draw[fill=white]  (v1) circle (7pt);
        \draw[fill=white]  (v2) circle (7pt);
        \draw[fill=white, dashed]  (v1err1) circle (7pt);
        \draw[fill=white, dashed]  (v1err2) circle (7pt);
        \draw[fill=white, dashed]  (v1err3) circle (7pt);
        
        \draw[fill=white]  (v3) circle (7pt);
        \draw[fill=white]  (v4) circle (7pt);
        \coordinate (triangleleft) at (4,-2);
        \coordinate (triangleright) at (8,-2);
        
        \fill[fill = white, draw=black] ([xshift={-7pt/sqrt(2)}, yshift={7pt/sqrt(2)}]v5) -- ([xshift={-7pt}]triangleleft) -- ([xshift=7pt]triangleright)
        -- ([xshift={7pt/sqrt(2)}, yshift={7pt/sqrt(2)}]v5) -- cycle;
        \node[fill=none, draw=none] (t) at (6,-1) {$L_v$};
        \draw[fill=white]  (v5) circle (7pt);
        \node[draw=none, fill=none] (some) at (v5) {$v$};
        \draw[line width=1pt, color=ForestGreen, ->, >=stealth] ([xshift={7/sqrt(2)}, yshift={7/sqrt(2)}] root) -> ([xshift={7/sqrt(2)}, yshift={7/sqrt(2)}]v1);
        \draw[line width=1pt, color=red, ->, >=stealth] ([xshift={7/sqrt(2)}, yshift={7/sqrt(2)}] v1) -> ([xshift={7/sqrt(2)}, yshift={7/sqrt(2)}]v1err1);
        
        \draw[line width=1pt, color=red, ->, >=stealth] ([xshift={7/sqrt(2)}, yshift={7/sqrt(2)}] v1err1) -> ([xshift={7/sqrt(2)}, yshift={7/sqrt(2)}]v1err2);
        
        \draw[line width=1pt, color=red, ->, >=stealth, dashed] ([xshift={-7/sqrt(2)}, yshift={-7/sqrt(2)}] v1err2) -> ([xshift={0}, yshift={-7*sqrt(2)}] v1err1);
    
        \draw[line width=1pt, color=red, ->, >=stealth] ([xshift={0}, yshift={-7*sqrt(2)}] v1err1) -> ([xshift={7/sqrt(2)}, yshift={-7/sqrt(2)}]v1err3);

        \draw[line width=1pt, color=red, ->, >=stealth, dashed] ([xshift={-7/sqrt(2)}, yshift={7/sqrt(2)}]v1err3) -> ([xshift={-7*sqrt(2)}, yshift={0}] v1err1);
        
        \draw[line width=1pt, color=red, ->, >=stealth, dashed] ([xshift={-7*sqrt(2)}, yshift={0}] v1err1) -> ([xshift={0}, yshift={-7*sqrt(2)}]v1);

        \draw[line width=1pt, color=ForestGreen, ->, >=stealth] ([xshift={0}, yshift={-7*sqrt(2)}]v1) -> ([xshift={7/sqrt(2)}, yshift={-7/sqrt(2)}] v2);
        \draw[line width=1pt, color=ForestGreen, ->, >=stealth] ([xshift={7/sqrt(2)}, yshift={-7/sqrt(2)}]v2) -> ([xshift={7*sqrt(2)}, yshift={0}] v3);
        \draw[line width=1pt, color=ForestGreen, ->, >=stealth] ([xshift={7*sqrt(2)}, yshift={0}] v3) -> ([xshift={7/sqrt(2)}, yshift={7/sqrt(2)}]v4);
        \draw[line width=1pt, color=ForestGreen, ->, >=stealth] ([xshift={7/sqrt(2)}, yshift={7/sqrt(2)}]v4) -> ([xshift={7/sqrt(2)}, yshift={7/sqrt(2)}]v5);
    \end{tikzpicture}}
    \caption{The picture represents the runtime of \cref{alg:noisy-tree}. The thick green path is $P'_{i,j}$ for some $i$ and $j$. The walk corresponding to the runtime of \cref{alg:noisy-tree} is represented with thin arrows: \textcolor{ForestGreen}{green} arrows represent \emph{good} rounds, \textcolor{red}{solid red} arrows represent bad rounds where protocol makes a mistake, and \textcolor{red}{dashed red} arrows represent bad rounds where protocol backtracks.}
    \label{fig:error-boosting-algorithm}
\end{figure}

Since \cref{alg:noisy-tree} performs $R$ rounds using in each round at most 2 instances of the randomized
\textsc{Equality} protocol with error $1/4$, the total amount of communication is at most $O(R) =
O(\max\{d, \log(1/\delta)\})$ as desired. Let us now verify correctness.

For any inputs $i, j$ there is a unique root-to-leaf path $P_{i,j}$ taken in $T$ ending at some leaf
$v$, and a corresponding unique path $P'_{i,j}$ in $T'$ which terminates at the subtree $L_v$. For
any execution of \cref{alg:noisy-tree}, we say a round $r$ is ``good'' if the pointer $c$ starts the
round on a vertex in $P'_{i,j} \cup L_v$ and also ends the round on a vertex in $P'_{i,j} \cup L_v$.
We say round $r$ is ``bad'' otherwise.  Write $\bm{g}$ for the number of good rounds and $\bm{b}$
for the number of bad rounds, which are random variables satisfying $R = \bm{b} + \bm{g}$.

\begin{claim}
If $\bm{g} > d$ then the protocol produces a correct output.
\end{claim}
\begin{proof}[Proof of claim]
Observe that, if $c \in P'_{i,j}$ at the start of round $r$, then the counter cannot move back up,
because the \textsc{Equality} protocol has one-sided error and will correctly report that the
concatenated strings are equal with probability 1.  So the protocol must have terminated with the
counter $c$ at a descendent of the $g^{th}$ node of $P'_{i,j}$. Since $\bm{g} > d$, the protocol
terminated with $c$ in the subtree of $T'$ that replaced the final node $v$ of $P_{i,j}$, meaning
that it will output the correct value.
\end{proof}

We say that the protocol \emph{makes a mistake} in round $r$ if the randomized \textsc{Equality}
protocol erroneously outputs ``equal'' in \cref{ntcheckabove} when $(a_{q_1}(i), \dotsc,
a_{q_t}(i)) \neq (b_{q_1}(j), \dotsc, b_{q_t}(j))$, or if these tuples are truly equal but the
protocol erroneously reports ``equal'' in \cref{ntcheckcurrent} when $a_c(i) \neq b_c(j)$.
Define the random variable $\bm{m}_r \define 1$ if the protocol makes a mistake in round $r$
and 0 otherwise, and define $\bm{m} = \sum_{r=1}^R \bm{m}_r$ for the total number of rounds where
the protocol makes a mistake.

\begin{claim}
$\bm{b} \leq 2\bm{m}$.
\end{claim}
\begin{proof}[Proof of claim]
Consider any bad round $r$. Either the counter $c$ moves up or down the tree. If the counter $c$
moves up to its parent $c'$, then we charge the bad round to the most recent round $r' < r$ where
the counter started at $c'$ and observe that the protocol must have made a mistake at round $r'$.
Otherwise, if the counter $c$ moves down the tree, we charge the bad round to $r$ itself and observe
that the protocol makes a mistake in round $r$. Then we see that each round where a mistake is made
is charged for at most $2$ bad rounds (one for itself, if the counter moves down; and one for the
earliest round where the counter returns to its current position).
\end{proof}

If the protocol outputs the incorrect value then we must have $\bm{g} = R-\bm{b} \leq d$ and therefore $R-d
\leq \bm{b} \leq 2\bm{m}$, so $\bm{m} \geq \frac{R-d}{2}$. It remains to bound the number of
mistakes $\bm{m}$; we write $\bm{m} = \sum_{r=1}^R \bm{m}_r$ where $\bm{m}_r$ indicates whether the
protocol makes a mistake in round $r$.

In any round $r$, conditional an all previous rounds, the probability that the protocol makes a
mistake is at most $1/4$: either there is an ancestor node in $P$ where a mistake was made in an
earlier round, in which case a mistake is made in round $r$ only if it makes an error in
\cref{ntcheckabove}; or the path $P$ is entirely correct and the protocol makes a mistake only
if there is an error in \cref{ntcheckcurrent}. So $\Pruc{}{\bm{m}_r = 1}{\bm{m}_1, \dotsc,
\bm{m}_{r-1}} \leq 1/4$ for every $r$ and $\mu \define \Ex{\bm{m}} \leq R/4$. Using known
concentration bounds (\eg Theorem 3.1 of \cite{IK10}), for any
$\tfrac{1}{4} \leq \gamma \leq 1$ we have $\Pr{ \bm{m} \geq \gamma R } \leq e^{-R \cdot D(\gamma \|
\delta) }$; in particular, since $R = 4 \cdot \max\{d, C\log(1/\delta)\}$, we have $\frac{R-d}{2}
\geq \frac{3R}{8}$, so for constant $\kappa \define D\left(\tfrac{3}{8} \| \tfrac{1}{4}\right) > 0$,
\[
  \Pr{ \bm{m} \geq \frac{R-d}{2} }
  \leq \Pr{ \bm{m} \geq \frac{3}{8} \cdot R } \leq e^{- R \cdot \kappa}
  \leq e^{- 4C\log(1/\delta) \cdot \kappa} \leq \delta \,,
\]
when we choose $C$ to be a sufficiently large constant.
\end{proof}


\section{No Better Boosting for Hamming Distance, and Consequences}

We now complete the proof of \cref{thm:intro-main-informal} by showing that
\cref{eq:intro-informal-badboosting} is tight for the \textsc{1-Hamming Distance} oracle.
We prove this with a direct-sum result, showing that computing $k$ independent copies of
\textsc{1-Hamming Distance} cannot be computed without the $\log k$-factor loss from boosting.
Let us define the direct sum problems.

For any function $f\colon X \times Y \to Z$ and any $k \in \bN$, we define function $f^{\otimes k}$
as the function which computes $k$ copies of $f$, \ie $f^{\otimes k}\colon X^k \times Y^k \to Z^k$
where on inputs $x \in X^k$ and $y \in Y^k$,
\[
  f^{\otimes k}(x,y) = (f(x_1, y_1), f(x_2, y_2), \dotsc, f(x_k, y_k)) \,.
\]
It is easy to see that
$\D^{\textsc{HD}_1}((\HD{1}{n})^{\otimes k}) = k$ for $n > 1$ since we can compute each copy of
$\HD{1}{n}$ with one query. In this section we prove:
\begin{boxtheorem}
\label{thm:no-boosting-for-hd1}
For all $n \geq 4k^2$, $\normalfont\R_{1/4}((\HD{1}{n})^{\otimes k}) = \Omega(k \log k)$.
Consequently, there exist matrices $M$ such that
\[
\normalfont\R_{1/4}(M)
  = \Omega\left( \D^{\HD{1}{}}(M) \cdot \log \D^{\HD{1}{}}(M)  \right) \,.
\]
\end{boxtheorem}

Our proof has two further consequences. The first is about randomized parity decision trees (see \eg
\cite{CGS21} for definitions and background on parity decision trees): it is not hard to see that
the randomized parity decision tree complexity of the \textsc{$1$-Hamming Weight} function
$\textsc{HW}_1^n\colon \{0,1\}^n \to \{0,1\}$ defined by $\textsc{HW}_1^n(x) = 1$ iff $|x| = 1$ is $\RPDT(\textsc{HW}_1^n) = O(1)$. Since $\HD{1}{n}(x,y) = \textsc{HW}_1^n(x\oplus y)$ and one can simulate each parity query with two bits of communication, we get
$\R_{1/4}((\HD{1}{n})^{\otimes k}) = O(\RPDT_{1/4}((\mathsf{HW}_1^n)^{\otimes k}))$.  Together these
statements imply:

\begin{boxcorollary}
\label{cor:intro-rpdt}
For $n \geq 4k^2$, 
$\normalfont\mathsf{RPDT}((\textsc{HW}_1^n)^{\otimes k}) = \Omega( k \log k )$.
\end{boxcorollary}

The second consequence of our proof, explained in \cref{section:eq-hd-separation}, is the optimal
$\Omega(\log n)$ lower bound on the number of \textsc{Equality} queries required to compute
$\HD{1}{n}$. All of these results come from our main lemma, a randomized reduction from the
\textsc{$k$-Hamming Distance} problem to $(\HD{1}{n})^{\otimes O(k)}$.  The \textsc{$k$-Hamming
Distance} communication problem is:

\begin{boxdefinition}[\textsc{$k$-Hamming Distance}]
For every $n, k \in \bN$, the \textsc{$k$-Hamming Distance} problem $\HD{k}{n} \colon \zo^n \to \zo$
is defined as
\[
  \HD{k}{n} = \begin{cases}
    1 &\text{ if } \dist(x,y) \leq k \\
    0 &\text{ otherwise.}
  \end{cases}
\]
\end{boxdefinition}
(Note that the randomized complexity $\R(\HD{k}{n})$ does not change by more than a constant factor
if we replace $\dist(x,y) \leq k$ with $\dist(x,y)=k$ in the definition, because each version of the
problem can be computed by at most 2 oracle queries to the other version.)

\subsection{Randomized Reduction Lemma}

\begin{boxlemma}
\label{lem:reduction-to-direct-sum}
Let $c = 9/10$. Then for all $k \in \bN$, $R = \log_{1/c} k$ and $\delta_0, \dots, \delta_R \in (0,1)$ such that $\sum_{i = 0}^R \delta_i \le 1/10$ we have
\[
\normalfont\R_{1/4}(\HD{k}{n})
  = O\left(\sum_{i = 0}^{R} \R_{\delta_i}((\HD{1}{n})^{\otimes (4k \cdot c^{i})}\right) .
\]
\end{boxlemma}

We thank Shachar Lovett for pointing out an improvement in this reduction which removed a factor
$\log\log \D^{\HD{1}{}}(M)$ from the denominator in \cref{thm:no-boosting-for-hd1}.

\begin{proof}
Our protocol for $\HD{k}{n}(x,y)$ is \cref{alg:reduction}. Let $c = 9/10$ and let $C$ be some
constant to be determined later.  For a string $x \in \zo^n$ and a set $S \subseteq [n]$, we will
write $x_S \in \zo^{|S|}$ for the substring of $x$ on coordinates $S$.

\begin{algorithm}
    \begin{algorithmic}[1]
    \Input{$x,y \in \{0,1\}^n$.}
    \State Initialize $T \gets [n]; \ell \gets k; j \gets 0.$
    \While{$\ell > C$} 
    \State Let $S_1, \dots, S_{4\ell}$ be a uniformly random partition of $T$.
    \State Let $u_i = x_{S_i};\; v_i = y_{S_i}$ be the substrings of $x,y$ on subsets $S_i$ for all $i \in [4\ell]$. 
    \State \label{invoke}Run $\delta_j$-error protocols for $(\HD{1}{n})^{\otimes
4\ell}\big((u_1, v_1), \dots, (u_{4\ell}, v_{4\ell})\big)$ and $\textsc{Eq}_n^{\otimes
4\ell}\big((u_1, v_1), \dots, (u_{4\ell}, v_{4\ell})\big)$. 
    \LComment{Assuming these subroutines are correct, we know $\dist(u_i, v_i)$ exactly,
    if $\dist(u_i,v_i) \in \zo$.}
    \State $w_i \gets \dist(u_i, v_i)$ if $\dist(u_i, v_i) \le 1$ and $2$ otherwise.
    \LComment{We can safely output 0 if we see more than $\ell$ differences:}
    \OneLineIf{$\sum_{i \in [4\ell]} w_i > \ell$}{{\bf return} 0.} \label{haltzero} 
    \LComment{In the next step, isolate the sets $S_i$ where the protocol finds exactly one difference.}
    \State{$s \gets |\{i \in [4\ell] \mid w_i = 1\}|$.}
    \LComment{If $\dist(x_T, y_T) > \ell$, we should see many sets with exactly one difference;
      output 1 otherwise:}
    \OneLineIf{$s < \ell/10$}{{\bf return} 1.} \label{haltone}
    \LComment{Throw out sets $S_i$ with at most one difference; update the number
$\ell$ of remaining differences.}
    \State $T \gets \bigcup_{i \in [4\ell]\colon w_i = 2} S_i$.
    \State $\ell \gets \ell - s$.
    \State $j \gets j + 1$.
    \EndWhile
    \State {\bf return} $\HD{\ell}{|T|}(x_T, y_T)$. 
    \end{algorithmic}
\caption{Hamming Distance Reduction}
\label{alg:reduction}
\end{algorithm}
First, let us calculate the cost of the protocol. As guaranteed by \cref{haltone}, at each
iteration the value of $\ell$ is reduced to at most $\tfrac{9}{10} \ell = c\ell$, so there are at
most $R = \log_{1/c} k$ iterations, and in the $i$-th iteration (indexed from zero), $\ell \le k c^i$.
Hence, at each iteration, the communication cost is at most
\[
\R_{\delta_i}((\HD{1}{n})^{\otimes 4 k c^i}) + \R_{\delta_i}((\EQ)^{\otimes 4k c^i})
\leq 2 \cdot \R_{\delta_i}((\HD{1}{n})^{\otimes 4 k c^i}) \,.
\]
Since $C$ is a constant, the cost of the final step with $\ell \leq C$ is $O(1)$.
    
Now let us estimate the error. Since there are at most $R = \log_{1/c} k$ iterations and the 2
protocols in \cref{invoke}~each have error at most $\delta_i$ in the $i$-th iteration, the total probability of an
error occurring in \cref{invoke}~is at most $2\cdot \sum_{i=0}^R \delta_i \le 1/5$. We may therefore assume from now on the
perfect correctness of the values $w_i$.

Under this assumption, the protocol maintains the invariant that the number of bits outside $T$
where $x,y$ differ is $\dist(x_{[n] \setminus T},
y_{[n]\setminus T}) = k - \ell$, so it cannot output the incorrect value in \cref{haltzero}. Let us
consider the probability that the protocol outputs the incorrect value in \cref{haltone}. This
only occurs if $\dist(x_T, y_T) > \ell$ and $s < \ell / 10$.  We need to estimate $\Pr{|\{i \in
[4\ell] \mid w_i = 1\}| \ge \ell/10}$. The size of the set $\{i \in [4\ell] \mid w_i > 0\}$ is the
number of unique colors we get when coloring each element $i$ of the set $\Delta_T \define \{ i \in T :
x_i \neq y_i \}$ of cardinality $|\Delta_T| = \dist(x_T, y_T)$ uniformly with color $\bm{\chi}_i
\sim [4\ell]$; call this number $\bm{\chi} \define |\{ \bm{\chi}_i : i \in \Delta_T\}|$.  We know that \cref{haltone}~does not halt, so $|\{i \in [4\ell]
\mid w_i = 2\}| < \ell/2$. Then, if $\bm{\chi} \geq (6/10)\ell$, it must be that $|\{ i \in [4\ell] : w_i = 1 \}|
\geq \bm{\chi} - \ell/2 \geq \ell/10$, so \cref{haltone}~does not halt. For
simplicity, since
$|\Delta_T| \geq \ell$, in the next expression we consider only the first $\ell$ elements of $\Delta_T$ and identify them with the set $[\ell]$. The probability we need
to estimate is 
    \begin{align*} \Pr{|\{\bm{\chi}_i \mid i \in [\ell] \}| \le 0.6\ell} &\le \sum_{S \in
\binom{[\ell]}{0.6\ell}} \Pr{\{\bm{\chi}_i \mid i \in [\ell]\} \subseteq \{\bm{\chi}_i \mid i \in S\}}\\
    &\le \binom{\ell}{0.6\ell} \cdot \left(\frac{6 }{10 \cdot 4}\right)^{0.4\ell} < 2^\ell \cdot 2^{\log_2 (3/20) \cdot 0.4 \ell} \le 2^{-.01 \ell}. \end{align*}
    We have that the total error is bounded by $\sum_{\ell=C}^\infty 2^{-.01 \ell} \le 2^{-.01 C} / (1-2^{-.01}) \le 100 \cdot 2^{-.01 C}$, so choosing $C$ to be large enough we get arbitrarily small constant error. 
\end{proof}

\subsection{Direct Sum Theorem for \texorpdfstring{\textsc{$1$-Hamming Distance}}{1 Hamming Distance}}

We require the lower bound on the communication cost of $\HD{k}{n}$:
\begin{boxtheorem}[\cite{Sag18}]
\label{thm:hd-lowerbound}
    For all $k^2 \leq \delta n$, $\normalfont\R_{\delta}(\HD{k}{n}) = \Omega(k \log(k/\delta))$.
\end{boxtheorem}
Now we can prove \cref{thm:no-boosting-for-hd1}.
\begin{proof}[Proof of \cref{thm:no-boosting-for-hd1}]
    Assume for contradiction that $\R_{1/4}((\HD{1}{n})^{\otimes k}) = o(k \log k)$, so
by standard boosting,
\[
  \forall \delta \in (0,1), \qquad \R_\delta((\HD{1}{n})^{\otimes k}) = o\left(k \log k \cdot \log
\frac{1}{\delta} \right).
\]
Then by \cref{lem:reduction-to-direct-sum}, with $c = 9/10$, $R = \log_{1/c} k$, and $\delta_i =
1/(200 \cdot c^i)$ for $i \in \{0, \dots, R\}$,
\begin{align*}
\R_{1/4}( \HD{k}{n} )
&= O\left( \sum_{i=0}^R \R_{\delta_i}( (\HD{1}{n})^{4kc^i} ) \right) 
= \sum_{i=0}^R o \left(  kc^i \log (kc^i) \log (200 \cdot c^i) \right) \\
&= \sum_{i=0}^R o(c^i \log (200 \cdot c^i) \cdot k \log k) \\
&= o(k \log k) \sum_{i=0}^R c^i \cdot i = o(k \log k) \,,
\end{align*}
which contradicts \cref{thm:hd-lowerbound} when $n \geq 4k^2$.
\end{proof}

Our \cref{cor:intro-rpdt} for randomized parity decision trees follows easily
from this theorem since a randomized parity decision tree for \textsc{1-Hamming Weight},
(or $k$ copies of it), can be simulated by a randomized communication protocol to
compute \textsc{1-Hamming Distance} (or $k$ copies of it).

\subsection{Lower Bound on Computing \texorpdfstring{$1$}{1}-Hamming Distance with \texorpdfstring{\textsc{Equality}}{Equality} Queries}
\label{section:eq-hd-separation}

Recently, \cite{HHH22dimfree,HWZ22} showed that \textsc{Equality} is not complete for the class
$\BPPZ$ of constant-cost communication problems, and \cite{FHHH24} showed that there is \emph{no}
complete problem for this class. The independent and concurrent proofs of \cite{HHH22dimfree,HWZ22}
both showed that $\D^\EQ(\HD{1}{n}) = \omega(1)$. We showed above that functions which reduce to
\textsc{Equality} have better boosting, while \textsc{$1$-Hamming Distance} does not, so
\textsc{$1$-Hamming Distance} cannot reduce to \textsc{Equality} -- this gives a new and unexpected
proof that \textsc{Equality} is not complete for $\BPPZ$:

\begin{boxcorollary}
$\normalfont\DEQ(\HD{1}{n}) = \omega(1)$. Therefore, \textsc{Equality} is not a complete problem for
$\BPPZ$.
\end{boxcorollary}

There is an easy upper bound of $\DEQ(\HD{1}{n}) = O(\log n)$ obtained using binary search.  With a
more careful argument we can strengthen the above result and get a new proof that this is optimal,
matching the lower bound already given in \cite{HHH22dimfree} by Fourier analysis.

\begin{boxtheorem}
\label{thm:hd-lb}
  $\normalfont \DEQ(\HD{1}{n}) = \Theta(\log n)$.
\end{boxtheorem}
\begin{proof}
Assume for the sake of contradiction that $\DEQ(\HD{1}{n}) = o(\log n)$, which immediately implies
$\DEQ((\HD{1}{n})^{\otimes k}) = o(k \log n)$. By \cref{thm:eq-betterboosting} we then have
$\R_{\delta}((\HD{1}{n})^{\otimes k}) \le o(k \log n) + O(\log 1/\delta)$. Applying
\cref{lem:reduction-to-direct-sum} we get, for $c = 9/10$ and $\delta = \delta_i = \frac{1}{11 \log_{1/c} k}$ for every $i\in \{0,1,\dots,R\}$,
\begin{align*}
  \R_{1/4}(\HD{k}{n})
    &= O\left( \sum_{i=0}^{\log_{1/c} k} \R_\delta( (\HD{1}{n})^{\otimes 4kc^i}) \right) \\
    &= \sum_{i=0}^{\log_{1/c} k} (o(kc^i \log n) + O(\log \log k)) = o(k \log n) + O(\log k \log\log k).
\end{align*}
Applying this inequality with $n = k^4$ we get $\R_{1/4}(\HD{k}{k^4}) = o(k \log k)$, which contradicts \cref{thm:hd-lowerbound}.
\end{proof}

\begin{remark}
\label{remark:additive}
It is interesting that the additive $O(\log(1/\delta))$ in \cref{thm:eq-betterboosting} is required
for this proof. If the $\log(1/\delta)$ term was multiplicative, we would get a bound of $o(k \log n
\cdot \log\log k)$ in the sum, giving $o(k \log k \log\log k)$ when we set $n = k^4$, which is not
in contradiction with \cref{thm:hd-lowerbound}. So the weaker (but still non-trivial) bound
$\R_{1/4}(M) = O(\DEQ(M))$ would not suffice, although it would still allow us to conclude
$\DEQ(\HD{1}{n}) = \omega(1)$. The trivial bound of $\R_{1/4}(M) = O(\DEQ(M) \log \DEQ(M))$ would
not allow us to prove even $\DEQ(\HD{1}{n}) = \omega(1)$.
\end{remark}

\section{Noisy-Tree Fails for Other Oracles}
\label{section:noisy-tree-fails}

At this point we cannot determine whether better boosting is possible \emph{only} for the
constant-cost protocols which reduce to \textsc{Equality}. But we can make some progress towards
this question by observing that the ``noisy-tree'' protocol in \cref{thm:eq-betterboosting} does not
work for any other oracles in $\BPPZ$. To state this formally, we must define a reasonable
generalization of that protocol.

The noisy-tree protocol relied on two properties of the \textsc{Equality} oracle. The first is that
it has one-sided error (the protocol for \textsc{Equality} will output the correct answer with
probability 1 when the inputs are equal). The second property is what we will call the
\emph{conjunction property}:

A query set $\cQ$ has the \emph{conjunction property} if there exists a constant $c$ such that for
all $d \in \bN$ and all $Q_1, \dotsc, Q_d \in \cQ$, $\D^\cQ\left( \bigwedge_{i=1}^d Q_i \right) \leq
c$ where $\bigwedge_{i=1}^d Q_i$ denotes the problem of computing
\[
  Q_1(x_1,y_1) \wedge Q_2(x_2,y_2) \wedge \dotsm \wedge Q_d(x_d,y_d) \,.
\]
on $d$ pairs of inputs $(x_1, y_1), \dotsc, (x_d, y_d)$. For example, \textsc{Equality} has
the conjunction property because computing
\[
  \Eq(x_1, y_1) \wedge \Eq(x_2, y_2) \wedge \dotsm \wedge \Eq(x_d, y_d)
\]
can be done with the single query $\Eq\left(((x_1, x_2, \dotsc, x_d), (y_1, y_2, \dotsc, y_d)
\right)$.
Following the proof of \cref{thm:eq-betterboosting}, we could claim the following result, which
would hold even for oracles $\cQ$ that have non-constant cost (but still using arbitrary-size oracle
queries\footnote{Arbitrary-size oracle queries may be sensible for non-constant cost problems that
still have bounded VC dimension, \eg \textsc{Greater-Than} oracles as in \cite{CLV19}.}):\\
\vspace{-1em}

\noindent
\textbf{``Theorem.''}\emph{
Let $\cQ$ be any query set satisfying the conjunction property, and whose elements $Q \in \cQ$ admit
one-sided error randomized communication protocols with cost $O(\R(\cQ))$. Then for any $M \in \zo^{N
\times N}$, $\R_\delta(M) = O( \D^\cQ(M) \cdot \R_{1/4}(\cQ) + \log\tfrac{1}{\delta})$.}\\

But it turns out that this does not really generalize \cref{thm:eq-betterboosting}, even if we
require only the conjunction property (\ie ignore one-sided error):

\begin{proposition}
\label{prop:no-generalization}
If $\cQ$ is a query set that satisfies the conjunction property, then it is either a subset of the
query set of \textsc{Equality}, or it is the set of all matrices.
\end{proposition}
To prove this, we use VC dimension. The VC dimension of a Boolean matrix $M$ is the largest $d$ such
that there are $d$ columns of $M$, where the submatrix of $M$ restricted to those columns contains
all $2^d$ possible distinct rows. A family $\cF$ of matrices has \emph{bounded VC dimension} if
there is a constant $d$ such that all $M \in \cF$ have VC dimension at most $d$. If $\cF$ is closed
under taking submatrices (and permutations), then it has bounded VC dimension if and only if it is
not the family of all matrices.

\begin{proof}[Proof of \cref{prop:no-generalization}]
\newcommand{\nand}{\left[ \begin{smallmatrix}1 & 1 \\ 1 &
0\end{smallmatrix}\right]}
If $\cQ$ does not contain the matrix $\nand$, then it is not hard to see that $\cQ$ is a subset of
the query set of \textsc{Equality}.  So we suppose that $\cQ$ contains the matrix $\nand$ and
satisfies the conjunction property. We first show:

\begin{claim}
Every matrix $M \in \zo^{N \times N}$ is a submatrix of $\bigwedge_{i=1}^N Q_i$ where each $Q_i =
\nand$.
\end{claim}
\begin{proof}[Proof of claim]
Let each $Q_i$ be a copy of $\nand$, so that $Q = \bigwedge_{i=1}^N Q_i$ has row space $[2]^N$ and
column space $[2]^N$. Let $M \in \zo^{N \times N}$ and map each row $x \in [N]$ of $M$ to the row
$v(x) \in [2]^N$ of $Q$ with
\[
  \forall j \in [N] : v(x)_j = \begin{cases}
    1 &\text{ if } M(x,j) = 1 \\
    2 &\text{ if } M(x,j) = 0 \,,
  \end{cases}
\]
and map each column $y \in [N]$ of $M$ to the column $w(y) \in [2]^N$ of $Q$ with
\[
  \forall j \in [N] : w(y)_j = \begin{cases}
    1 &\text{ if } j \neq y \\
    2 &\text{ if } j = y \,.
  \end{cases}
\]
For any row $x$ and column $y$ of $M$, if $M(x,y) = 1$ then
\[
  Q_i( v(x)_j, w(y)_j ) = \begin{cases}
    Q_i(1, 1) = 1 &\text{ if } M(x,j) = 1 \text{ and } j \neq y \\
    Q_i(1, 2) = 1 &\text{ if } M(x,j) = 1 \text{ and } j = y \\
    Q_i(2, 1) = 1 &\text{ if } M(x,j) = 0 \text{ and } j \neq y \,.
  \end{cases}
\]
This covers all the cases, since we never have $M(x,j) = 0$ and $j = y$, so $Q(v(x),w(y)) = 1$.
Finally, if $M(x,y) = 0$ then 
\[
  Q_i( v(x)_y, w(y)_y ) = Q_i( 2, 2 ) = 0 \,,
\]
so $Q(v(x), w(y)) = 0$. Therefore $M$ is a submatrix of $Q$.
\end{proof}

By the conjunction property, there is a constant $c$ such that $\D^\cQ(M) \leq
\D^\cQ(\bigwedge_{i=1}^N Q_i) \leq c$ for all $M \in \zo^{N \times N}$. Therefore, there is
a constant $C$ and a function $f\colon \zo^C \to \zo$ such that all matrices $M$ can be written as
\[
  M(x,y) \define f(Q_1(x,y), Q_2(x,y), \dotsc, Q_C(x,y))
\]
where each $Q_i \in \cQ$ (think of $f$ as the function which simulates the protocol for $\D^\cQ(M)$
using the answers to each query $Q_i$; see \eg \cite{HZ24} for the simple proof of this fact). Let
$f(\cQ)$ denote the set of all matrices which can be achieved in this way, which we have argued is
the set of all matrices. For the sake of contradiction, assume that $\cQ$ is not the set of all
matrices, so that the VC dimension $\mathsf{VC}(\cQ)$ is bounded. Then standard VC dimension
arguments (see \eg \cite{Mat13}) show that the VC dimension of $f(\cQ)$ is at most $O(
\mathsf{VC}(\cQ) \cdot C \log C)$.  Since $C$ is constant, the VC dimension of $f(\cQ)$ is therefore
also bounded, but $f(\cQ)$ contains all matrices, so this is a contradiction and $\cQ$ must contain
all matrices.
\end{proof}

\begin{remark}
If one is interested only in constant-cost oracles, we may replace the conjunction property
$\D^\cQ(\bigwedge_i Q_i) \leq c$ with the property $\R_{1/4}\left( \bigwedge_{i=1}^d Q_i \right) =
O(1)$, but the same proof rules out this generalization as well.
\end{remark}

\section*{Acknowledgments}

We thank Mika G\"o\"os and Eric Blais for discussions on the topic of this article, and Lianna
Hambardzumyan for kind remarks on an early draft. We thank Shachar Lovett for helping us remove a
$\log\log(\cdot)$ from the denominators in \cref{thm:no-boosting-for-hd1,thm:hd-lowerbound}.  Both authors are supported by Swiss
State Secretariat for Education, Research, and Innovation (SERI) under contract number MB22.00026.
Nathaniel Harms is also supported by  NSERC Postdoctoral Fellowship.

\bibliographystyle{alpha}
\bibliography{references.bib}
\end{document}

%% file: header.tex
\usepackage[dvipsnames]{xcolor}
\usepackage{tikz}
 \usetikzlibrary{positioning}
 \usetikzlibrary[patterns, patterns.meta]
 \usetikzlibrary[fadings]
 \usetikzlibrary{decorations.pathreplacing}
\usepackage{bm, bbm}
\usepackage[letterpaper, portrait, margin=1in]{geometry}
\usepackage[colorlinks=true,linkcolor=blue,citecolor=ForestGreen]{hyperref}
\usepackage{algorithm}
\usepackage[noend]{algpseudocode}
\usepackage{url}
\usepackage{amsmath,amssymb,amsthm}
\usepackage{thmtools,thm-restate}
\usepackage[noabbrev,capitalise,nameinlink]{cleveref}
\usepackage{mathtools}
\usepackage{xspace}
\usepackage{verbatim}
\usepackage{mathrsfs}
\usepackage{tabularx}
\usepackage{enumitem}
\usepackage{derivative}
\usepackage{multirow}
\usepackage{diagbox}
\usepackage{nicematrix}
\usepackage[most]{tcolorbox}

\usepackage{dsfont}

\newcommand*\ie{i.\kern.1em e.\ }
\newcommand*\eg{e.\kern.1em g.\ }
\newcommand*\cf{c.\kern.1em f.\ }
\newcommand*\almev{a.\kern.1em e.\ }

\definecolor{ama-iro}{RGB}{0, 158, 243.0}
\definecolor{fuyu-gaki}{RGB}{251, 74, 52}
\definecolor{momiji}{RGB}{245, 70, 111}
\definecolor{hotaru-bi}{RGB}{229,221,58} 
\definecolor{kon-peki}{RGB}{1,120,217}
\definecolor{shin-kai}{RGB}{77,98,152}
\definecolor{shin-ryoku}{RGB}{1,145,97}
\definecolor{yama-budo}{RGB}{171,14,122}

\definecolor{citecolor}{HTML}{5E9100}

\definecolor{theoremcolor}{HTML}{FFE1D9}
\definecolor{resultcolor}{HTML}{FFE1D9}

\definecolor{constraintcolor}{HTML}{BBD49B}
\definecolor{goalcolor}{HTML}{CAE4A7}

\definecolor{remarkcolor}{HTML}{E3EEC7}

\definecolor{definitioncolor}{HTML}{FCDFBE}

\definecolor{examplecolor}{HTML}{F9E9D9}
\definecolor{questioncolor}{HTML}{DDE8EB}

\definecolor{captioncolor}{RGB}{128,128,128}

\newcommand{\nathan}[1]{{\color{ama-iro} \textbf{NH:} #1}}

\hypersetup{
  colorlinks=true,
  linkcolor=momiji,
  filecolor=ama-iro,
  urlcolor=momiji,
  citecolor=citecolor
}

\theoremstyle{plain}
\newtheorem{theorem}{Theorem}[section]
\newtheorem{lemma}[theorem]{Lemma}

\newtheorem{proposition}[theorem]{Proposition}
\newtheorem{claim}[theorem]{Claim}
\newtheorem{corollary}[theorem]{Corollary}
\newtheorem{question}[theorem]{Question}
\newtheorem{question*}{Question}
\newtheorem{conjecture}[theorem]{Conjecture}

\newtheorem{observation}[theorem]{Observation}
\newtheorem{assumption}[theorem]{Assumption}

\crefname{claim}{Claim}{Claims}
\crefname{fact}{Fact}{Facts}

\theoremstyle{definition}

\newtheorem{definition}[theorem]{Definition}
\newtheorem{remark}[theorem]{Remark}
\newtheorem{example}[theorem]{Example}
\newtheorem{goal}[theorem]{Goal}
\newtheorem{goal*}{Goal}

\theoremstyle{plain}

\newenvironment{boxtheorem}{\begin{theorem}}{\end{theorem}}
\tcolorboxenvironment{boxtheorem}{colback=theoremcolor, colframe=white,
    colbacktitle=theoremcolor, coltitle=theoremcolor}

\newenvironment{boxlemma}{\begin{lemma}}{\end{lemma}}
\tcolorboxenvironment{boxlemma}{colback=resultcolor, colframe=white,
    colbacktitle=resultcolor, coltitle=resultcolor}

\tcolorboxenvironment{boxproposition}{colback=resultcolor, colframe=white,
    colbacktitle=resultcolor, coltitle=resultcolor}

\newenvironment{boxcorollary}{\begin{corollary}}{\end{corollary}}
\tcolorboxenvironment{boxcorollary}{colback=resultcolor, colframe=white,
    colbacktitle=resultcolor, coltitle=resultcolor}

\tcolorboxenvironment{boxobservation}{colback=resultcolor, colframe=white,
    colbacktitle=resultcolor, coltitle=resultcolor}

\newenvironment{boxquestion}{\begin{question}}{\end{question}}
\tcolorboxenvironment{boxquestion}{colback=questioncolor, colframe=white,
    colbacktitle=questioncolor, coltitle=shin-kai}
\newenvironment{boxquestion*}{\begin{question*}}{\end{question*}}
\tcolorboxenvironment{boxquestion*}{colback=questioncolor, colframe=white,
    colbacktitle=questioncolor, coltitle=questioncolor}

\newenvironment{boxdefinition}{\begin{definition}}{\end{definition}}
\tcolorboxenvironment{boxdefinition}{colback=definitioncolor, colframe=white,
    colbacktitle=definitioncolor, coltitle=definitioncolor}

\tcolorboxenvironment{boxassumption}{colback=definitioncolor, colframe=white,
    colbacktitle=definitioncolor, coltitle=definitioncolor}

\newenvironment{boxexample}{\begin{example}}{\end{example}}
\tcolorboxenvironment{boxexample}{colback=examplecolor, colframe=white,
    colbacktitle=examplecolor, coltitle=examplecolor}

\newtheorem{exercise}[theorem]{Exercise}

\tcolorboxenvironment{boxexercise}{colback=exercisecolor, colframe=white,
    colbacktitle=exercisecolor, coltitle=exercisecolor}

\tcolorboxenvironment{boxgoal}{colback=goalcolor, colframe=white,
    colbacktitle=goalcolor, coltitle=goalcolor}

\newenvironment{boxgoal*}{\begin{goal*}}{\end{goal*}}
\tcolorboxenvironment{boxgoal*}{colback=goalcolor, colframe=white,
    colbacktitle=goalcolor, coltitle=goalcolor}

\tcolorboxenvironment{boxremark}{colback=remarkcolor, colframe=white,
    colbacktitle=remarkcolor, coltitle=remarkcolor}

\theoremstyle{plain}

\newcounter{constraintcounter}
\renewcommand{\theconstraintcounter}{%
  \ifnum\value{constraintcounter}=5 IVb\else \Roman{constraintcounter}\fi
}

\newtheorem{constraint}[constraintcounter]{Constraint}
\newtheorem{constraint*}{Constraint}
\crefname{constraint}{Constraint}{Constraints}

\tcolorboxenvironment{boxconstraint}{colback=constraintcolor, colframe=white,
    colbacktitle=constraintcolor, coltitle=constraintcolor}

\newenvironment{boxconstraint*}{\begin{constraint*}}{\end{constraint*}}
\tcolorboxenvironment{boxconstraint*}{colback=constraintcolor, colframe=white,
    colbacktitle=constraintcolor, coltitle=constraintcolor}

\newcommand{\ignore}[1]{}

\DeclareMathOperator{\poly}{poly}

\newcommand{\dist}{\mathsf{dist}}

\newcommand{\Ex}[1]{\bE \left[ #1 \right]}

\renewcommand{\Pr}[1]{\bP \left[ #1 \right]} 

\newcommand{\Pruc}[3]{\underset{ #1 }\bP \left[ #2 \;\; \mid \;\; #3 \right]}

\newcommand*{\define}{\mathrel{\vcenter{\baselineskip0.5ex \lineskiplimit0pt
                      \hbox{\scriptsize.}\hbox{\scriptsize.}}}%
                      =}

\newcommand{\zo}{\{0,1\}}

\newcommand{\cF}{\ensuremath{\mathcal{F}}}

\newcommand{\cM}{\ensuremath{\mathcal{M}}}

\newcommand{\cP}{\ensuremath{\mathcal{P}}}
\newcommand{\cQ}{\ensuremath{\mathcal{Q}}}

\newcommand{\bN}{\ensuremath{\mathbb{N}}}
\newcommand{\bP}{\ensuremath{\mathbb{P}}}